\documentclass[a4paper, final]{article}
\usepackage{a4wide}
\usepackage[utf8]{inputenc}
\usepackage{amsfonts}
\usepackage{amsmath, amsthm}
\usepackage{amssymb} 
\usepackage{enumitem}
\usepackage{microtype}
\usepackage{tikz}
\usepackage{hyperref}

\theoremstyle{definition}
\newtheorem{theorem}{Theorem}
\newtheorem{lemma}[theorem]{Lemma}
\newtheorem{proposition}[theorem]{Proposition}
\newtheorem{corollary}[theorem]{Corollary}
\newtheorem{example}[theorem]{Example}

\newcommand{\tenum}[2]{#1_1,\ldots,#1_#2}

\newcommand{\disjun}{\uplus}

\newcommand{\V}{\mathcal{V}}

\newcommand{\T}{\mathcal{T}}
\newcommand{\MF}{\mathcal{F}}
\newcommand{\MR}{\mathcal{P}}

\newcommand{\sat}{\vDash}

\newcommand{\imp}{\rightarrow}

\newcommand{\I}[1]{\mathcal{I}_{#1}}

\newcommand{\vars}{\operatorname{vars}}
\newcommand{\tatoms}{\operatorname{ta}}
\newcommand{\rrs}{\operatorname{rrs}}

\newcommand{\shortrules}[6]{\noindent\begin{minipage}{#6ex}{\bfseries #1}\end{minipage} $\;$ #2 $\;\Rightarrow_{\operatorname{#5}}\;$ #3 \par\smallskip\noindent #4} 
\newcommand{\rfl}{\text{ref}}
\newcommand{\irr}{\text{irr}}

\newcommand{\cA}{\mathcal{A}}
\newcommand{\cH}{\mathcal{H}}

\newcommand{\fUA}{\mathsf{A}}
\newcommand{\fUH}{\mathsf{H}}

\newcommand{\<}{\langle}
\renewcommand{\>}{\rangle}

\newcommand{\fa}{\mathsf{a}}
\newcommand{\fb}{\mathsf{b}}
\newcommand{\fc}{\mathsf{c}}
\newcommand{\fd}{\mathsf{d}}
\newcommand{\fe}{\mathsf{e}}

\begin{document}

\title{On the Expressivity and Applicability of Model Representation Formalisms}

\author{
	\begin{tabular}{l}
		Andreas Teucke\\
		\small\textit{Max Planck Institute for Informatics, Saarland Informatics Campus, Saarbr\"ucken, Germany}
	\end{tabular}
	\and
	\begin{tabular}{l}
		Marco Voigt\\
		\small\textit{Max Planck Institute for Informatics, Saarland Informatics Campus, Saarbr\"ucken, Germany,}\\
		\small\textit{Saarbr\"ucken Graduate School of Computer Science}
	\end{tabular}
	\and
	\begin{tabular}{l}
		Christoph Weidenbach \\
		\small\textit{Max Planck Institute for Informatics, Saarland Informatics Campus, Saarbr\"ucken, Germany}
	\end{tabular}
}	
\date{}
\maketitle

\begin{abstract}
  A number of first-order calculi employ an explicit model representation formalism
  for automated reasoning and for detecting satisfiability.
  Many of these formalisms can represent infinite Herbrand models.
  The first-order fragment of monadic, shallow, linear, Horn (MSLH)
  clauses, is such a formalism used in the approximation refinement calculus.
  Our first result is a finite model property for MSLH
  clause sets. Therefore, MSLH clause sets cannot represent models
  of clause sets with inherently infinite models.
  Through a translation to tree automata, we further show  that this limitation
  also applies to the linear fragments of implicit generalizations,
  which is the formalism used in the model-evolution calculus,
  to atoms with disequality constraints, the formalisms used in the
  non-redundant clause learning calculus (NRCL), and
  to atoms with membership constraints, a formalism used for example in
  decision procedures for algebraic data types.
  Although these formalisms cannot represent models
  of clause sets with inherently infinite models, through an additional
  approximation step they can. This is our second main result.
  For clause sets including the definition of an equivalence relation
  with the help of an additional, novel approximation, called reflexive relation splitting,
  the approximation refinement calculus can automatically
  show satisfiability through the MSLH clause set formalism.
\end{abstract}

%%%%%%%%%%%%%%%%%%%%%%%%%%%%%%%%%%%%%%%%%%%%%%
%%%%%%%%%%%%%%%%%%%%%%%%%%%%%%%%%%%%%%%%%%%%%%
%%%%%%%%%%%%%%%%%%%%%%%%%%%%%%%%%%%%%%%%%%%%%%
%%%%%%%%%%%%%%%%%%%%%%%%%%%%%%%%%%%%%%%%%%%%%%

\section{Introduction} \label{sec:intro}

Proving satisfiability of a first-order clause set is more difficult
than proving unsatisfiability, in general. Still, for many applications
the detection of failing refutations by establishing a counter model is more than desirable. In the
past, several methods, calculi and systems have been presented that can detect
satisfiability of a clause set, in particular, if there is a finite model that is not too large. 
The approaches can be separated into the following classes: 
\begin{enumerate}[label=(\arabic*)]
\item the model building is integrated into a first-order calculus or a decision procedure forsome fragment, directly operating on
   the first-order clause set, complete for unsatisfiability, e.g.,~\cite{CaferraLeitschPeltier04,BaumgartnerFuchsTinelli06,PiskacEtAl10,SuterEtAl2010,AlagiWeidenbach15,BonacinaEtAl15,BonacinaPlaisted16}, \label{apro1}
\item the model building is integrated into a first-order calculus that operates on the
   first-order clause set modulo an approximation, complete for unsatisfiability, e.g.,~\cite{Korovin13,TeuckeW16,HernandezKorovin17}, \label{apro2}
\item the model building aims at finding finite models without being 
      complete for unsatisfiability, e.g.,~\cite{Slaney94,McCune2003,ClaessenSoerensson03}, \label{apro3}
\item the model building aims at finding finite and infinite models without being 
      complete for unsatisfiability, e.g.,~\cite{Peltier09}, \label{apro4}
\end{enumerate}
where superposition~\cite{BachmairGanzinger94b} does not belong to any of the above classes, because the model building
is implicit and reached by a finite saturation of the clause set modulo inferences and the elimination
of redundant clauses.

The approaches in classes \ref{apro1} and \ref{apro2}
select inferences with respect to the explicit (partial) model by identifying a false clause
(instance). Therefore, the representation of models needs to be effective, e.g., falsity of a clause (instance)
with respect to the model needs to be (efficiently) decidable. This is not the case for superposition,
where this problem is undecidable.  The clause set consisting of $R(\epsilon,\epsilon)$ and 
clauses  $R(x,y) \imp R(t_i[x], s_i[y])$
where the $t_i, s_i$ are terms build over the monadic functions $g, h$ and variables $x$, $y$, respectively,
is saturated with strictly
maximal atoms $R(t_i[x], s_i[y])$ and can encode the words of a PCP~\cite{Post46} 
$({\tt w}_1,\ldots,{\tt w}_n)$, $({\tt v}_1,\ldots,{\tt v}_n)$ over letters $g, h$. 
That means words are represented by nestings of monadic functions.
The PCP has a solution iff a ground atom $R(g(t),g(t))$ or $R(h(t),h(t))$ is a consequence of the above clause set. 
This corresponds to testing whether one of the clauses $R(g(x),g(x))$ or $R(h(x),h(x))$ has
a false instance with respect to the model generated by the superposition model operator. 

Reasoning with respect to a (partial) model assumption has advantages. The superposition completeness proof
shows that an inference with a clause that is false in the current partial model is not redundant~\cite{BachmairGanzinger94b}.
This has meanwhile also been shown for the CDCL~\cite{Weidenbach15} and the NRCL~\cite{AlagiWeidenbach15} calculus.
The non-redundant inference property might also hold for other calculi of classes \ref{apro1} and \ref{apro2}.

Our first contribution is showing that the model representation used in~\cite{TeuckeW16}, monadic shallow linear Horn clauses 
(MSLH) has the finite model property, Section~\ref{sec:moprop}. This means that
if a finite MSLH clause set has a model, it also has a finite model. Hence, MSLH clause sets cannot be used
directly to represent models of clause sets with inherently infinite models.
A further consequence is that any calculus in class~\ref{apro1}, where the model representation
can be represented by an MSLH clause set, cannot terminate on clause sets with inherently infinite models.
A more detailed discussion of this aspect is contained in Section~\ref{sec:modrepform}. 

The fact that MSLH clause sets have the finite model property does not mean that the approximation refinement (AR) calculus
presented in \cite{TeuckeW16} cannot be used for finding infinite models of clause sets with inherently infinite models.
The reason is that the MSLH model representation in \cite{TeuckeW16} does not directly relate
to the clause set, but via an approximation. For the approximation it is shown in a constructive
way that it preserves satisfiability. This is done modulo the minimal Herbrand model of a saturated
MSLH clause set. Such Herbrand models become infinite as soon as there are non-constant function symbols.
So the question is whether AR can actually terminate on
clause sets with inherently infinite models. 
In Section~\ref{sec:mofind}, we show that this is the
case for certain classes of such clause sets relying on reflexivity of a binary (equivalence) relation.
The technique we propose is an additional approximation called \emph{reflexive relation splitting}.
A similar relationship between a clause set and its approximation was 
already observed in \cite{Peltier09} where an approximation of a first-order clause
set into a class of tree automata is used in order to find finite and infinite models.

Our results concerning the MSLH fragment and the reflexive relation splitting modulo
the AR calculus can be demonstrated by the following example.
Consider the following three clauses defining a reflexive binary relation $R$
(see~\cite{CaferraLeitschPeltier04}, page~55 for further discussion of this example). \newline
\centerline{$\{R(x,x), \quad R(g(x),g(y)) \rightarrow R(x,y), \quad \neg R(g(x),c)\}$}
The clause set has only infinite models. No resolution inference between $R(x,x)$ and $\neg R(g(x),c)$
is possible. Following the AR approach~\cite{TeuckeW16}, the MSLH clause set \newline
\centerline{$\{T(f_R(x,y)), \quad T(f_R(g(x),g(y))) \rightarrow T(f_R(x,y)), \quad \neg T(f_R(g(x), c))\}$}
is generated. We write unit clauses as single literals, and non-unit clauses as implications.
The relation $R$ is translated into a binary function $f_R$ over a monadic predicate $T$. 
The approximation is the replacement of $R(x,x)$ by $T(f_R(x,y))$,
where now the connection between the non-linear occurrences of $x$ is lost. As a consequence, a refutation
containing a resolution step between $T(f_R(x,y))$ and $\neg T(f_R(g(x), c))$ with substitution
$\{x\mapsto g(v), y\mapsto c\}$
is possible, which cannot be lifted to the original clause set because $g(v)$ and $c$ are not unifiable. The refinement then
excludes this particular instance by generating $R(g(x),g(x))$, however, after approximating
this clause, the empty clause can be derived again.
This time the derivation also uses the second clause,
where the substitution instance of the refutation contains one further nesting of $g$.
The approximation refinement approach does not terminate on this example.

Obviously, if in the approximation the inference between $T(f_R(x,x))$ and $\neg T(f_R(g(x), c))$ can be blocked,
saturation will terminate without finding a contradiction. As said, in the original clause set this
inference is not possible, because of the non-linear occurrence of $x$. Now the idea is to split the relation
$R$ into its reflexive and irreflexive part, denoted by the two predicates $R_\rfl$ and $R_\irr$, respectively.
The original clause set is satisfiable if and only if the following clause set is satisfiable\newline
\centerline{$\{R_\rfl(x,x), \quad R_\irr(g(x),g(y)) \rightarrow R_\irr(x,y), \quad \neg R_\irr(g(x),c)\}$,}
details are explained in Section~\ref{sec:mofind}. After approximation it becomes\newline
\centerline{
\begin{tabular}{cr}
	$\{T(f_{R_\rfl}(x,y))^*, T(f_{R_\irr}(g(x),g(y)))^+ \rightarrow T(f_{R_\irr}(x,y)), \neg T(f_{R_\irr}(g(x), c))^*\}$ & ($\dagger$)\\
  \end{tabular}}
 \newline
where $^*$ highlights maximal and $^+$ selected literals of the ordered resolution calculus
used to decide MSLH clauses~\cite{TeuckeW15,TeuckeW16}. 
There are no possible inferences generating further clauses, i.e.\ the set is already saturated.

The infinite minimal Herbrand model is $\I{} = \{T(f_{R_\rfl}(g^i(c), g^j(c))) \mid i,j \geq 0\}$
which is also a model for the clause set before approximation~\cite{TeuckeW15,TeuckeW16} by
simply undoing the shift of $R_\irr$, $R_\rfl$ to the function level: 
$\I{} = \{R_\rfl(g^i(c), g^j(c))) \mid i, j \geq 0\}$. 
Nestings of functions in the Herbrand model representing relations, e.g., $f_{R_\rfl}$,  
can be prevented by adding further MSLH clauses. We omit these here for simplicity.
This model can then be translated,
see the proof of Lemma~\ref{lem:reflrelsplit}, into the Herbrand model 
$\I{} = \{R(g^i(c), g^i(c))) \mid i \geq 0\}$
of the original clause set.

In Section~\ref{sec:moprop}, we prove a finite model for saturated, satisfiable MSLH clause sets.
For the example, see ($\dagger$), the thus constructed model has the 
domain $\fUA := \{ \fa_c, \fa^{(1)}, \fa^{(2)}, \fa^{(3)}, \fb^{(1)},$ $\fb^{(2)}, \fb^{(3)} \}$.
The predicate $T$ is interpreted with the set $\{ \fb^{(1)}, \fb^{(2)}, \fb^{(3)} \}$.
For the constant $c$ we use the distinguished element $\fa_c$.
The interpretation of the function $f_{R_\rfl}$ is given in the following function table: \newline
\centerline{$ 
	\begin{array}{c@{\hspace{5ex}}r@{\;}c@{\;}c@{\hspace{3ex}}l}
			& \< \fa_c, \fa_c \> 		&\longmapsto	& \fb^{(1)} \\
			& \< \fa^{(i)}, \fa^{(i)} \>	&\longmapsto	& \fb^{(j)} 	& \text{for every $i$ and some $j \neq i$} \\
			& \< \fb^{(i)}, \fb^{(i)} \> 	&\longmapsto	& \fb^{(j)} 	& \text{for every $i$ and some $j \neq i$} \\
			& \< \fc, \fd \> 			&\longmapsto	& \fa^{(j)} 	& \text{for any $\fc, \fd \in \fUA$ with $\fc \neq \fd$ and} \\
			&					&			&		& \text{some $j$ chosen such that for any $i$,} \\
			&					&			&		& \text{if $\fc$ or $\fd$ is equal to $\fa^{(i)}$ or $\fb^{(i)}$, then $j \neq i$}
	\end{array}
$}
For the function $f_{R_\irr}$ we get a similar function table in which every pair $\<\fc,\fd\>$ 
is mapped to some $\fa^{(j)}$, where $j$ is chosen such that $\fc, \fd \neq \fa^{(j)}$.
Finally, the interpretation of the function $g$ is given by 
	$g(\fa_c) = \fa^{(1)}$, 
	$g(\fa^{(i)}) = \fa^{(j)}$ and 
	$g(\fb^{(i)}) = \fa^{(j)}$ for every $i$ and some $j \neq i$.

The paper is now organized as follows: after fixing some notions and notations, Section~\ref{sec:prelim},
the finite model property of MSLH clause sets is shown in Section~\ref{sec:moprop}. Consequences of this result
for other model representation formalisms are discussed in Section~\ref{sec:modrepform}. In Section~\ref{sec:mofind}
reflexive relation splitting is introduced and its application to AR investigated.
The present paper ends with a discussion on the obtained results and future research directions, Section~\ref{sec:discussion}.

%%%%%%%%%%%%%%%%%%%%%%%%%%%%%%%%%%%%%%%%%%%%%%
%%%%%%%%%%%%%%%%%%%%%%%%%%%%%%%%%%%%%%%%%%%%%%
%%%%%%%%%%%%%%%%%%%%%%%%%%%%%%%%%%%%%%%%%%%%%%
%%%%%%%%%%%%%%%%%%%%%%%%%%%%%%%%%%%%%%%%%%%%%%

\section{Preliminaries} \label{sec:prelim}

We consider a standard first-order language without equality where
letters $v,w,$ $x,y,z$ denote variables, $f,g,h$ functions, $a,b,c$ constants, $s,t$ terms, and Greek letters 
$\sigma,\tau,\rho$ are used for substitutions. 
$S,P,Q,R$ denote predicates, $A,B$ atoms, $E,K,L$ literals, $C,D$ clauses, $N$ clause sets  and $\V$ sets of variables.
The notation $[\neg] A$ denotes $A$ or its negation.
The signature $\Sigma=(\MF,\MR)$ consists of two disjoint, non-empty, in general infinite sets of  function and predicate symbols
$\MF$ and $\MR$, respectively.
The set of all \emph{terms} over the variables in $\V$ is $\T(\MF,\V)$.
If there are no variables, then terms, literals and clauses are called \emph{ground}, respectively.
A \emph{substitution} $\sigma$ is denoted by pairs $\{x \mapsto t\}$.
A substitution $\sigma$ is a \emph{grounding} substitution for a term, atom, literal, clause
if the application of $\sigma$ yields a ground term, ground atom, ground literal, ground clause, respectively.

The set of \emph{free} variables of an atom $A$ (term $t$, literal $L$, clause $C$) is denoted by $\vars(A)$ 
($\vars(t), \vars(L), \vars(C)$).
A predicate with exactly one argument is called \emph{monadic}. 
A term is \emph{complex} if it is not a variable and \emph{shallow} if it is a constant, a variable,
or of the form $f(\tenum{x}{n})$.
A term, atom is called \emph{linear} if there are no duplicate variable occurrences. 

A \emph{clause} is a multiset of literals which we write as an implication $\Gamma \imp \Delta$ 
where the atoms in the multiset $\Delta$ (the\emph{ succedent}) denote the positive literals and the atoms in the multiset $\Gamma$ (the \emph{antecedent}) the negative literals. Alternatively, we write a clause also as a disjunction of its literals.
We write $\square$ for the empty clause.
We abbreviate disjoint set union with sequencing, for example,  we write 
$\Gamma,\Gamma' \imp \Delta,L$ instead of $\Gamma \cup \Gamma' \imp \Delta \cup \{L\}$.
A clause $\Gamma \imp \Delta$ is called an \emph{MSLH} clause, if 
(i)~$\Delta$ contains at most one atom, i.e., the clause is Horn,
(ii)~all occurring predicates are monadic,
(iii)~the argument of any monadic atom in $\Delta$ is shallow and linear.
The first-order fragment consisting of finite MSLH clause sets we call \emph{MSLH}. 

An \emph{atom ordering} $\prec$ is an irreflexive, well-founded, total ordering on ground atoms.  
It is lifted to literals by defining $A \prec \neg A \prec B$
for any atoms $A$, $B$ with $A\prec B$. It is lifted to clauses by its multiset extension.
The ordering is lifted from the ground level through ground instantiation: for two different atoms $A$, $B$
containing variables, $A\prec B$ if $A\sigma \prec B\sigma$ for all grounding substitutions $\sigma$ and
the atoms are incomparable otherwise. A literal $L$ is \emph{maximal} (\emph{strictly maximal}) in a clause $C\lor L$ if there is
no literal $K\in C$ with $L\prec K$ ($L\preceq K$).
The clause ordering is compatible with the atom ordering; 
if the maximal atom in $C$ is greater than the maximal atom in $D$ then $D \prec C$.
We use  $\prec$ simultaneously to denote an atom ordering and its multiset, literal, and clause extensions.
For a ground clause set $N$ and clause $C$, the set $N^{\prec C}=\{D \in N \mid D \prec C\}$ 
denotes the clauses of $N$ smaller than $C$.

As usual, we interpret atoms, clauses, and clause sets 
with respect to \emph{structures} $\cA$, also called \emph{interpretations}, consisting of a nonempty universe $\fUA$ 
and interpretations $c^\cA$, $f^\cA$, and $P^\cA$ of all occurring constants, functions, and predicates. 
We often use a special kind of interpretations, called \emph{Herbrand interpretations}, 
whose universe is the set of all ground terms.
A \emph{Herbrand interpretation} $\I{}$ is represented by a -- possibly infinite -- set of ground atoms.
A ground atom $A$ is \emph{true} in $\I{}$ if $A\in\I{}$ and \emph{false}, otherwise.
$\I{}$ is said to \emph{satisfy} a ground clause $C= \Gamma \imp \Delta$, denoted by $\I{}\sat C$, 
if   $\Delta \cap \I{} \neq \emptyset$ or $\Gamma \not\subseteq \I{}$.
A non-ground clause $C$ is satisfied by $\I{}$ if $\I{}\sat C\sigma$ for every grounding substitution $\sigma$.
An interpretation $\I{}$ is called a \emph{model} of $N$, $\I{}\sat N$, if $\I{}\sat C$ for every $C\in N$. 
A Herbrand model $\I{}$ of $N$ is considered \emph{minimal} (with respect to set inclusion) if there is no model $\I{}'$ 
with $\I{}'\subset \I{}$ and $\I{}'\sat N$.
A set of clauses $N$ is \emph{satisfiable}, if there exists a model that satisfies $N$. 
Otherwise, the set is \emph{unsatisfiable}.

The superpositon calculus~\cite{BachmairGanzinger94b} restricted to first-order logic without equality results in the ordered
resolution calculus together with the superposition redundancy criterion and partial model operator.
For ordered resolution, a selection function is assumed that may select negative literals in clauses.
Then $(C\lor D)\sigma$ is an ordered resolution inference  between a clause $C\lor A$ and a clause $D\lor \neg B$,
if (i)~$\sigma$ is the mgu between $A$ and $B$, (ii)~$A\sigma$ is strictly maximal in $(C\lor A)\sigma$
and nothing is selected in $C\lor A$, (iii)~$\neg B\sigma$ is maximal in $(D\lor \neg B)\sigma$ or selected.
The clause $(C\lor A)\sigma$ is an ordered factoring inference on a clause $C\lor A \lor A'$,
if (i)~$\sigma$ is the mgu between $A$ and $A'$, (ii)~$A\sigma$ is maximal in $(C\lor A)\sigma$
and nothing is selected in $C\lor A\lor A'$. 
Selection is stable under instantiation, i.e., if $\neg A$ is selected in $\neg A \lor C$ it is also selected
in $(\neg A \lor C)\sigma$, for any substitution $\sigma$.
A clause $C$ is \emph{redundant} with respect to a clause set $N$, if for all ground instances $C\sigma$
there are ground instances $D_1\sigma_1,\ldots,D_n\sigma_n$, $\{D_1,\ldots,D_n\}\subseteq N$,
$D_i\sigma_i \prec C\sigma$ for all $i$, such that $D_1\sigma_1,\ldots,D_n\sigma_n\models C\sigma$, i.e.,
$C\sigma$ is implied by smaller ground instances from clauses in $N$.
A clause set $N$ is called \emph{saturated} if all clauses generated by ordered resolution or
ordered factoring from clauses in 
$N$ are either redundant or contained in $N$.
Given a ground clause set $N$ and an ordering $\prec$ we can construct a (partial) Herbrand model $N_{\cal I}$
for $N$ by the superposition partial model operator inductively as follows:
\[
\renewcommand{\arraystretch}{1.3}
\begin{array}{rcl}
N_C & := &\bigcup_{D\prec C} \delta_D \\
\delta_D & := & \left\{\begin{array}{ll} 
\{P(\tenum{t}{n})\} & \mbox{if} \;D= D'\lor P(\tenum{t}{n}), P(\tenum{t}{n}) \;\mbox{strictly maximal,}\\
                    & \mbox{no literal selected in}\;D\; \mbox{and}\;N_D\not\models D\\
\emptyset & \mbox{otherwise}\\ \end{array}\right.\\
N_{\cal I} & := & \bigcup_{C\in N} \delta_C\\
\end{array}
\]
Clauses $C$ with $\delta_C \not = \emptyset$ are called \emph{productive}. For a non-ground
clause set $N$ we define $N_{\cal I} := (\{C\sigma \mid C\in N, \sigma \text{ grounding for $C$}\})_{\cal I}$.
The main completeness result of superposition is: for a clause set $N$ let $N^*$ be its (possibly infinite)
saturation, then either $\square\in N^*$ and $N$ is unsatisfiable, or $N^*_{\cal I} \models N$~\cite{BachmairGanzinger94b}.

%%%%%%%%%%%%%%%%%%%%%%%%%%%%%%%%%%%%%%%%%%%%%%
%%%%%%%%%%%%%%%%%%%%%%%%%%%%%%%%%%%%%%%%%%%%%%
%%%%%%%%%%%%%%%%%%%%%%%%%%%%%%%%%%%%%%%%%%%%%%
%%%%%%%%%%%%%%%%%%%%%%%%%%%%%%%%%%%%%%%%%%%%%%

\section{MSLH Model Properties}\label{sec:moprop}

By definition, Herbrand models for MSLH clause sets with non-constant function symbols have an infinite domain.
In what follows we show how to construct non-Herbrand models with finite domains for satisfiable finite MSLH clause sets.
The constructed model is a finite representation of the \emph{minimal Herbrand model}.

%%%%%%%%%%%%%%%%%%%%%%%%%%%%%%%%%%%%%%
Consider a satisfiable finite MSLH clause set $N$.
It is known that $N$ can be finitely saturated using superposition (ordered resolution)
with an appropriate ordering and selection strategy such that the following property holds for the obtained saturated clause set $N^*$~\cite{Weidenbach99cade}.
Every clause $C$ in $N^*$ that is productive in the sense of the superposition model operator has the form 
{$C =\; P_1(x_1), \ldots, P_n(x_n) \rightarrow S(f(\tenum{y}{m}))$}
where $\{x_1,\ldots,x_n\}\subseteq \{y_1,\ldots,y_m\}$,\linebreak $f(\tenum{y}{m})$ is linear, and
$S(f(\tenum{y}{m}))$ is strictly maximal in $C$. 
Such a saturation can, e.g., be obtained by choosing for $\prec$ a Knuth-Bendix-Ordering (KBO) 
with weight one for all function symbols, variables, and a selection
strategy that selects a negative literal $P_i(t_i)$ in any clause $P_1(t_1), \ldots,$ $P_n(t_n) \rightarrow S(f(\tenum{y}{m}))$
if $t_i$ is not a variable, if $t_i$ is a variable that does not occur in $f(\tenum{y}{m})$, or if $t_i$ is a variable
in a clause $P_1(x), \ldots, P_n(x) \rightarrow S(x)$~\cite{Weidenbach99cade,TeuckeW17}.

%%%%%%%%%%%%%%%%%%%%%%%%%%%%%%%%%%%%%%
\begin{proposition}[Entailed by Lemma~4 from~\cite{Weidenbach99cade}]\label{proposition:SaturationMSLHclauseSets}
	Consider a satisfiable finite MSLH clause set $N$.
	There is a finite MSLH clause set $N^*$ such that $N \subseteq N^*$ and $N \models N^*$ and
	there is a (minimal) Herbrand model $\cH \models N^*$ such that for every ground  atom $A$ of the form $S\bigl( f(s_1, \ldots, s_m)\bigr)$ we have $\cH \models A$ only if there is some clause $C$ in $N^*$ and a variable assignment $\beta$ with the following properties:
	\begin{enumerate}[label=(\alph{*}), ref=(\alph{*})]
		\item \label{enum:SaturationMSLHclauseSets:I}
			$C$ has the form $P_1(x_1), \ldots, P_n(x_n) \imp S\bigl( f(y_1, \ldots, y_m) \bigr)$ where $\{x_1,\ldots,x_n\}\subseteq \{y_1,\ldots,y_m\}$, the $y_1, \ldots, y_m$ are pairwise distinct, and $m,n \geq 0$;
		\item \label{enum:SaturationMSLHclauseSets:II}
			we have $\beta(y_i) = s_i$ for every $i$, $1 \leq i \leq m$; and
		\item  \label{enum:SaturationMSLHclauseSets:III}
			we have $\cH, \beta \models P_j(x_j)$ for every $j$, $1 \leq j \leq n$.
	\end{enumerate}	
\end{proposition}

Since $N^*$ is satisfiable and all its clauses are Horn\index{formula!Horn~$\sim$}, it possesses a unique minimal Herbrand model $\cH$\index{Herbrand model!minimal~$\sim$} (cf.~\cite{Ebbinghaus1994}, Chapter~XI, Theorem~3.8).
The property described in Proposition~\ref{proposition:SaturationMSLHclauseSets} provides the key to construct a finite model for $N$ and $N^*$ from $\cH$.
The following example is intended to illustrate the ideas underlying the construction in a simplified form.
\begin{example}\label{example:MSLHclauseSet}
	Consider the following set of MSLH clauses with constants $a$ and $b$:
		\begin{align*}
			N := \{& P(a),\; Q(b),\quad \neg P(z) \vee \neg Q(z) \vee \neg R(z) \\
				 &\neg P(u) \vee \neg P(u') \vee P(f(u,u')), \quad
				 \neg Q(v) \vee \neg Q(v') \vee Q(f(v,v')), \\
				& \neg P(x) \vee R(f(x,y)), \quad
				 \neg P(y) \vee R(f(x,y)), \\
				& \neg Q(x) \vee R(f(x,y)), \quad
				 \neg Q(y) \vee R(f(x,y))  \} ~.
		\end{align*}
	The set $N$ is satisfied by the minimal Herbrand interpretation $\cH$ with
	\begin{align*}
		P^\cH &:= \bigl\{ a, f(a,a),  f\bigl(a, f(a,a)\bigr), f\bigl(f(a,a), a\bigr),  f\bigl(f(a,a), f(a,a)\bigr), \ldots \bigr\} ~, \\
		Q^\cH &:= \bigl\{ b, f(b,b), f\bigl(b, f(b,b)\bigr), f\bigl(f(b,b), b\bigr), f\bigl(f(b,b), f(b,b)\bigr), \ldots \bigr\} ~, \\
		R^\cH &:= \bigl\{ f(s,t) \bigm| s \in P^\cH \text{ or } t \in Q^\cH \bigr\} ~.
	\end{align*}
	The interpretation $\cH$, together with $N^* := N$, satisfies the conditions of Proposition~\ref{proposition:SaturationMSLHclauseSets}:
	for every term $f(s,t)$ that belongs to $R^{\cH}$ we have that one of the clauses $\neg P(x) \vee R(f(x,y))$ or $\neg P(y) \vee R(f(x,y))$ or $\neg Q(x) \vee R(f(x,y))$ or $\neg Q(y) \vee R(f(x,y))$ enforces $\cH \models R(f(s,t))$ because of $\cH \models P(s)$ or $\cH \models P(t)$ or $\cH \models Q(s)$ or $\cH \models Q(t)$, respectively.
	Similarly, the presence of any term $f(\ldots)$ in $P^{\cH}$ or $Q^{\cH}$ is enforced by one of the clauses $\neg P(u) \vee \neg P(u') \vee P(f(u,u'))$ or $\neg Q(v) \vee \neg Q(v') \vee Q(f(v,v'))$.
	
	These requirements towards the minimality of $\cH$ provide us with a certain knowledge about distinct terms $f(s,t)$ and $f(s',t')$.
	Suppose the terms $s$ and $s'$ are indistinguishable with respect to their membership in $P^{\cH}, Q^{\cH}, R^{\cH}$.
	Further suppose that the same holds for the terms $t$ and $t'$.
	Then, $f(s,t)$ and $f(s',t')$ are also indistinguishable with respect to their belonging to $P^{\cH}, Q^{\cH}$, and $R^{\cH}$, because the arguments $s, t$ and $s', t'$ trigger the same productive clauses.
	A formal statement of this property is given in Lemma~\ref{lemma:congruenceArgumentColors}.
	
	Based on this observation, we use $\cH$ as a blueprint for a finite model $\cA$, which is depicted in Figure~\ref{figure:ExemplaryModelOfMSLHclauseSet}.
	The domain of $\cA$ shall be $\fUA := \{\fa, \fb, \fc, \fd, \fe\}$, and we set $a^\cA := \fa$ and $b^\cA := \fb$.
	The predicate symbols are interpreted by
	$
		P^\cA := \{ \fa, \fc \} , 
		Q^\cA := \{ \fb, \fd \} , 
		R^\cA := \{ \fc, \fd, \fe\}
	$.
	Moreover, we define \\
		\centerline{$
			\begin{array}{l@{\hspace{5ex}}l@{\hspace{5ex}}l@{\hspace{5ex}}l}
				f^\cA(\fa, \fa) := \fc &
				f^\cA(\fa, \fc) := \fc &
				f^\cA(\fc, \fa) := \fc &
				f^\cA(\fc, \fc) := \fc \\
				f^\cA(\fb, \fb) := \fd &
				f^\cA(\fb, \fd) := \fd &
				f^\cA(\fd, \fb) := \fd &
				f^\cA(\fd, \fd) := \fd .
			\end{array}
			$}
	For all other inputs, $f^\cA$ shall yield $\fe$ as output.
	Every domain element in $\fUA$ represents one equivalence class of the terms in $\cH$'s Herbrand domain\index{Herbrand domain} with respect to membership in the sets $P^{\cH}$, $Q^{\cH}$, and $R^{\cH}$.
	The domain element $\fa$ represents the class $[a] := \{ a \}$ of terms that belong to $P^{\cH}$ and to no other set.
	Similarly, $\fb$ represents $[b] := \{ b \}$ of terms that belong to $Q^{\cH}$ and to no other set.
	The element $\fc$ represents the class of all terms belonging to $P^{\cH} \cap R^{\cH}$, i.e.\ to the class containing $f(a,a), f(a,f(a,a))$ and so on.
	The class of terms belonging to $Q^{\cH} \cap R^{\cH}$ is represented by $\fd$.
	Finally, $\fe$ corresponds to the class of all terms that are member of $R^{\cH}$ but of none of the other predicates, e.g.\ $f(a,b), f(a, f(b,a))$.
	
	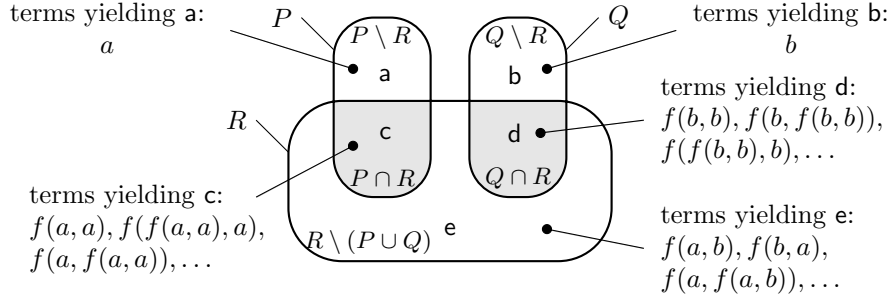
\begin{figure}
		\centering
			\begin{tikzpicture}[scale=0.85]
				\begin{scope}		
					\clip (0,0) rectangle (5,2.5);
					\filldraw[thick, fill=gray!20!white, draw=white] 
							[rounded corners = 15pt] (0.7,1) rectangle (2.2,3.8)
							[rounded corners = 15pt] (2.8,1) rectangle (4.3,3.8);
				\end{scope}		
				\draw[thick, draw=black] 
						[rounded corners = 15pt] (0.7,1) rectangle (2.2,3.8)
						[rounded corners = 15pt] (2.8,1) rectangle (4.3,3.8)
						[rounded corners = 20pt] (0,0) rectangle (5,2.5);
				\draw 		(0.8,3.5) node[anchor=west]{\small$P \setminus R$}
						(2.9,3.5) node[anchor=west]{\small$Q \setminus R$}
						(0.8,1.3) node[anchor=west]{\small$P \cap R$}
						(2.9,1.3) node[anchor=west]{\small$Q \cap R$}
						(0.1,0.3) node[anchor=west]{\small$R\setminus (P \cup Q)$};
				\draw 		(1.5,2.9) node{$\fa$}
						(3.5,2.9) node{$\fb$}
						(1.5,2) node{$\fc$}
						(3.5,2) node{$\fd$}
						(2.5,0.5) node{$\fe$};
						
				\draw [draw=black]	
						(0,1.7)	-- (-0.5,2.2)	node[anchor=east] {$R$}
						(0.7,3.3)	-- (0.2,3.8)	node[anchor=east] {$P$}
						(4.3,3.3)	-- (4.8,3.8)	node[anchor=west] {$Q$};		
									
				\filldraw [black]	
				 		(1,3)	circle	(2pt) -- (-0.9,3.6)	node[anchor=east] {\!\!\begin{tabular}{c} terms yielding $\fa$:\\ $a$ \end{tabular}}
						(4,3)	circle	(2pt) -- (6,3.6)	node[anchor=west] {\!\!\begin{tabular}{c} terms yielding $\fb$:\\ $b$ \end{tabular}}
						(1,1.8)	circle	(2pt) -- (-0.5,1)	node[anchor=north east] {\!\!\begin{tabular}{l} \strut\\[-5ex] terms yielding $\fc$:\\ $f(a,a), f(f(a,a),a),$\\ $f(a, f(a,a)), \ldots$ \end{tabular}\hspace{-3ex}\strut}
						(3.9,2)	circle	(2pt) -- (5.5,2.2)	node[anchor=west] {\!\!\begin{tabular}{l} terms yielding $\fd$:\\ $f(b,b), f(b,f(b,b)),$\\ $f(f(b,b), b), \ldots$ \end{tabular}}
						(4,0.5)	circle	(2pt) -- (5.5,0.2)	node[anchor=west] {\!\!\begin{tabular}{l} terms yielding $\fe$:\\ $f(a,b), f(b,a),$\\ $f(a, f(a,b)), \ldots$ \end{tabular}};
			\end{tikzpicture}
		\caption{Illustration of the model $\cA$ of $N$ from Example~\ref{example:MSLHclauseSet}.}
		\label{figure:ExemplaryModelOfMSLHclauseSet}
	\end{figure}
\end{example}
Next, we describe formally how to construct a finite model for the given satisfiable and finite MSLH clause set $N$.
Let $N^*$ and $\cH$ be the objects described in  Proposition~\ref{proposition:SaturationMSLHclauseSets}.
Then, we have $\cH \models N^*$ and $\cH \models N$.
Let $\fUH$ be the domain of $\cH$, i.e.\ $\fUH$ is the set of all ground terms over the vocabulary underlying $N$.
We aim at constructing a finite model $\cA \models N$ starting from $\cH$.

Let $\Pi$ denote the set of all predicates occurring in $N$, and recall that $\Pi$ contains only unary predicate symbols.
Let $\sim$ be the equivalence relation on $\fUH$ such that $s \sim t$ holds if and only if we have for every $P \in \Pi$ that $\cH \models P(s)$ if and only if $\cH \models P(t)$.

%%%%%%%%%%%%%%%%%%%%%%%%%%%%%%%%%%%%%%%%%%%%%%%%%%%%%%
\begin{lemma}\label{lemma:congruenceArgumentColors}
	For every non-constant function symbol $f$ in $N$ of arity $m$ and all tuples $\<s_1, \ldots, s_m\>,$ $\<t_1, \ldots, t_m\> \in \fUH^m$ for which $s_i \sim t_i$ holds for every $i$ we have 
	$f(s_1, \ldots, s_m) \sim f(t_1, \ldots, t_m)$.
\end{lemma}
\begin{proof}
	By Definition of $\cH$, $\cH \models S\bigl(f(s_1, \ldots, s_m)\bigr)$ entails that there is a clause $C$ of the form $\neg P_1(x_1) \vee \ldots \vee \neg P_n(x_n) \vee S\bigl( f(y_1, \ldots, y_m) \bigr)$ in $N^*$ and a variable assignment $\beta$ that satisfy Properties~\ref{enum:SaturationMSLHclauseSets:I} to~\ref{enum:SaturationMSLHclauseSets:III} from Proposition~\ref{proposition:SaturationMSLHclauseSets}.
	Let $\gamma$ be a variable assignment for which we have $\gamma(y_i) := t_i$ for every $i$.
	Notice that such a $\gamma$ with $\bigl\< \gamma(y_1), \ldots, \gamma(y_m) \bigr\> = \<t_1, \ldots, t_m\>$ always exists because the $y_1, \ldots, y_m$ are pairwise distinct.
	Since we assume $s_i \sim t_i$ for every $i$ and because of $\{x_1, \ldots, x_n\} \subseteq \{y_1, \ldots, y_m\}$, Conditions~\ref{enum:SaturationMSLHclauseSets:II} and~\ref{enum:SaturationMSLHclauseSets:III} of Proposition~\ref{proposition:SaturationMSLHclauseSets} stipulate for every $j$ that $\beta(x_j) \in P_j^\cH$ and, hence, we also have $\gamma(x_j) \in P_j^\cH$.
	Since $\cH$ is a model of $N^*$, we have $\cH, \gamma \models C$.
	This together with $\cH, \gamma \models P_j(x_j)$, for every $j$, entails $\cH, \gamma \models S(f(y_1, \ldots, y_m))$.
	Put differently, we have $\cH \models S\bigl( f(t_1, \ldots, t_m) \bigr)$.
	
	Consequently, for every $S$ we observe that $\cH \models S\bigl( f(s_1, \ldots, s_m) \bigr)$ entails $\cH \models S\bigl( f(t_1, \ldots, t_m) \bigr)$.
	The converse direction can be shown by a symmetric argument.
\end{proof}

We now construct the finite structure $\cA$.
The universe of $\cA$ shall be $\fUA := \fUH/_\sim := \bigl\{ [s]_\sim \mid s \in \fUH \bigr\}$, where $[s]_\sim$ denotes the (unique) equivalence class with respect to $\sim$ which contains the term $s$.
For every function symbol $f$ (including constants) we set $f^\cA \bigl([s_1]_\sim, \ldots, [s_m]_\sim\bigr) := \bigl[ f(s_1, \ldots, s_m) \bigr]_\sim$ for all ground terms $s_1, \ldots, s_m$.
Finally, we define each predicate $P$ under $\cA$ by $P^\cA := \bigl\{ [s]_\sim \mid \cH \models P(s) \bigr\}$.

%%%%%%%%%%%%%%%%%%%%%%%%%%%%%%%%%%%%%%%%%%%%%%%%%%%%%%
\begin{lemma}\label{lemma:termColor}
	Let $\gamma$ be any variable assignment over $\cA$'s domain.
	Let $\beta$ be some variable assignment over $\cH$'s domain defined such that for every $x$ we have $\gamma(x) = \bigl[\beta(x)\bigr]_\sim$.
	By definition of $\cH$, such a $\beta$ must exist.
	Then, for every term $t$ in $N$ and every predicate $P$ we have $\cA, \gamma \models P(t)$ if and only if $\cH, \beta \models P(t)$.
\end{lemma}
\begin{proof}
	We proceed by case distinction regarding the structure of the term $t$.
	If $t = x$ is a variable, then we have
		$\cA, \gamma \models P(x)$
	if and only if
		$\gamma(x) = \bigl[\beta(x)\bigr]_\sim \in P^\cA$	
	if and only if
		$\beta(x) \in P^\cH$ 
	if and only if
		$\cH, \beta \models P(x)$.
	If $t = c$ is a constant, then we have
		$\cA, \gamma \models P(c)$
	if and only if
		$c^\cA = [c]_\sim \in P^\cA$	
	if and only if 	
		$c \in P^\cH$
	if and	only if 
		$\cH, \beta \models P(c)$.
	Suppose $t = f(s_1, \ldots, s_m)$ for some function $f$ of arity $m \geq 1$ and terms $s_1, \ldots, s_m$.
	Let $t_1, \ldots, t_m$ be ground terms such that $\cA(\gamma)(s_i) = [t_i]_\sim$. 
	Such terms exist by definition of $\cH$.
	Then, we have
		$\cA, \gamma \models P\bigl( f(s_1, \ldots, s_m) \bigr)$
	if and only if
		$f^\cA \bigl( [t_1]_\sim, \ldots, [t_m]_\sim \bigr) = \bigl[f(t_1, \ldots, t_m)\bigr]_\sim \in P^\cA$
	if and only if
		$\cH \models P \bigl( f(t_1, \ldots, t_m) \bigr)$.
	A straightforward induction on the structure of the terms $s_i$ yields $t_i \sim \cH(\beta)(s_i)$ for every $i$ (see below).
	Then, by Lemma~\ref{lemma:congruenceArgumentColors}, we have 
		$\cH \models P \bigl( f(t_1, \ldots, t_m) \bigr)$
	if and only if
		$\cH \models P \bigl( f \bigl( \cH(\beta)(s_1), \ldots, \cH(\beta)(s_m) \bigr) \bigr)$
	if and only if
		$\cH,\beta \models P \bigl( f(s_1, \ldots, s_m) \bigr)$.
	
	It remains to prove 	$t_i \sim \cH(\beta)(s_i)$ for any $i$ or, equivalently, $A(\gamma)(s_i) = \bigl[ \cH(\beta)(s_i) \bigr]_\sim$.
	We show the latter by induction on the structure of $s_i$.
	If $s_i = y$, we have $\cA(\gamma)(y) = \gamma(y) = [\beta(y)]_\sim = \bigl[ \cH(\beta)(y) \bigr]_\sim$.
	If $s_i = c$, we have $\cA(\gamma)(c) = c^\cA = [c]_\sim = \bigl[ \cH(\beta)(c) \bigr]_\sim$.
	Let $s_i = g(r_1, \ldots, r_k)$.
	Then, we have
		\begin{align*}
			\cA(\gamma)(g(r_1, \ldots, r_k) 
				&= 
					g^\cA \bigl( \cA(\gamma)(r_1), \ldots, \cA(\gamma)(r_k) \bigr) \\
				&\stackrel{\text{IH}}{=} 
					g^\cA \bigl( \bigl[ \cH(\beta)(r_1) \bigr]_\sim, \ldots, \bigl[ \cH(\beta)(r_k) \bigr]_\sim \bigr) \\
				&=
					\bigl[ g \bigl( \cH(\beta)(r_1), \ldots, \cH(\beta)(r_k) \bigr) \bigr]_\sim \\
				&=
					\bigl[ \cH(\beta)\bigl( g( r_1, \ldots, r_k) \bigr) \bigr]_\sim ~. \qedhere
		\end{align*}	
\end{proof}	

%%%%%%%%%%%%%%%%%%%%%%%%%%%%%%%%%%%%%%%%%%%%%%%%%%%%%%%
For the special case of ground terms, there is a simpler form of Lemma~\ref{lemma:termColor}:
\begin{corollary}\label{coro:gatomsfm}
	For every ground term $t$ and every predicate symbol $P$ we have
		$\cA \models P(t)$ if and only if $\cH \models P(t)$.
\end{corollary}
%
%%%%%%%%%%%%%%%%%%%%%%%%%%%%%%%%%%%%%%%%%%%%%%%%%%%%%%
Using Lemma~\ref{lemma:termColor}, it is easy to show that $N$ is satisfied by the finite structure $\cA$.
\begin{lemma}\label{lemma:FiniteModelSatisfiesClauses}
	$\cA$ is a model of $N$.
\end{lemma}
\begin{proof}
	Let $C$ be any clause in $N$.
	Since $N$ is satisfiable, $C$ cannot be the empty clause.
	Suppose there is some variable assignment $\gamma$ over $\cA$'s domain such that $\cA, \gamma \not\models C$.
	Let $\beta$ be some variable assignment over $\cH$'s domain defined such that $\gamma(x) = [\beta(x)]_\sim$ for every $x$.
	Such a $\beta$ must exist.
	Consider any atom $P(t)$ in $C$.
	By Lemma~\ref{lemma:termColor}, we have $\cA, \gamma \models P(t)$ if and only if $\cH, \beta \models P(t)$.
	But then, $\cA, \gamma \not\models C$ entails $\cH, \beta \not\models C$.
	This contradicts our assumption that $\cH$ is a model of $N$.
	Hence, we must have $\cA \models C$.
\end{proof}

%%%%%%%%%%%%%%%%%%%%%%%%%%%%%%%%%%%%%%%%%%%%%%%%%%%%%%
\begin{theorem}[Finite Model Property for MSLH]\label{theorem:FiniteModelPropertyForMSLH}
	Every satisfiable finite MSLH clause set $N$ has a finite model whose domain contains at most $2^p$ elements, where $p$ is the number of predicates in $N$.
\end{theorem}

%%%%%%%%%%%%%%%%%%%%%%%%%%%%%%%%%%%%%%%%%%%%%%
%%%%%%%%%%%%%%%%%%%%%%%%%%%%%%%%%%%%%%%%%%%%%%
%%%%%%%%%%%%%%%%%%%%%%%%%%%%%%%%%%%%%%%%%%%%%%
%%%%%%%%%%%%%%%%%%%%%%%%%%%%%%%%%%%%%%%%%%%%%%

\section{Model Representation Formalisms}\label{sec:modrepform}

Many known explicit first-order model representation formalisms are build on
sequences of literals, often enhanced with constraints, eventually representing
Herbrand models, e.g.,~\cite{BaumgartnerFuchsTinelli06,PiskacEtAl10,BonacinaEtAl15,AlagiWeidenbach15,BonacinaPlaisted16}, 
so called constraints atomic representations (CARMs)~\cite{CaferraLeitschPeltier04}.
A thorough discussion of all known CARM model representation formalisms is
beyond the scope of this paper. We concentrate on three basic building blocks
of known model representation formalisms: atoms with
disequality constraints~\cite{CaferraLeitschPeltier04,Comon91} (ADCs), implicit
generalizations~\cite{FermuellerPilcher07} (IGs) and atoms with membership constraints~\cite{CaferraLeitschPeltier04,tata2007} (AMCs).
They form the basis for a number of concrete model representation formalisms that actually
appear in the above mentioned calculi.
For this section we consider a fixed, finite signature $\Sigma=(\MF,\MR)$, e.g., the function and predicate
symbols occurring in some finite clause set $N$. The results in this section for all three model representations
will be the same: if terms, literals are linear, the models represented by
the respective approaches have the finite model property. We will prove this as follows:
(i)~we provide an effective linear time translation of atoms with disequality constraints to implicit generalizations;
(ii)~we provide a linear time translation of implicit generalizations to intersections of tree automata~\cite{tata2007} or complements thereof;
(iii)~we represent an atom with membership constraints by a tree automaton.
Then, because tree automata are closed under intersection and complement, potentially causing an exponential
blow up in size~\cite{tata2007},  the atoms generated by ADCs and IGs can also be represented by the accepted language of a single tree automaton.
The accepted language of a tree automaton can be represented by a finite MSLH clause set, e.g., see~\cite{JacquemardMeyerEtAl98}.
Thus, by Theorem~\ref{theorem:FiniteModelPropertyForMSLH}, linear ADCs, linear IGs, and linear AMCs have all the finite model property, i.e.,
they cannot represent models for clause sets with inherently infinite models.

A \emph{linear} ADC has the form $(A\colon x_1\neq t_1,\ldots, x_n\neq t_n)$ where the $x_i$ are all
different and occur in $A$, the $x_i$ do not occur in any $t_j$, the variables of the $t_j$ do not occur
in $A$ and $A$ as well as all $t_j$ are linear. The ground atoms generated by such an ADC
are all ground atoms $A\sigma$ such that there is no $\delta$ with $x_i\sigma = t_i\delta$ for some $i$.
A \emph{linear} IG
is an expression $A/\{B_1,\ldots,B_n\}$ where $A$ and the $B_i$ are linear atoms. Every ground instance of $A$ that is not an instance
of any $B_i$ is generated by the IG $A/\{B_1,\ldots,B_n\}$.
The ground atoms generated by an ADC are exactly the ground atoms generated by the respective linear IG
$A/\{A\{x_1\mapsto t_1\},\ldots,A\{x_n\mapsto t_n\}\}$.

A tree automaton consists of a finite set ${\cal Q}$ of states, a finite set ${\cal O}$ of operators,
a subset of accepting states ${\cal Q}_A\subseteq {\cal Q}$,
and a finite set of rules $f(q_1\ldots,q_n) \mapsto q$ where $q, q_i\in {\cal Q}$, $f\in {\cal O}$. The accepted language of a tree
automaton is inductively defined by $f(t_1,\ldots,t_n)\in q^{\cal A}$ if there is a rule $f(q_1\ldots,q_n) \mapsto q$
and $t_i\in q_i^{\cal A}$ for all $i$. The overall accepted language is then $\bigcup \;\{q^{\cal A}\mid q\in {\cal Q}_A\}$.

For example, the ground instances of the linear atom $R(x, g(a, y))$ over signature $\Sigma = (\{g,a,b\}, \{R\})$ is
the accepted language of the tree automaton ${\cal O} = \{R ,g, a, b\}$ with rules $a \mapsto q_1$, $b \mapsto q_1$,
$g(q_1,q_1) \mapsto q_1$, hence state $q_1$ accepts all ground terms, $a \mapsto q_2$, $g(q_2,q_1) \mapsto q_3$, and
$R(q_1,q_3) \mapsto q_4$ where $q_4$ is the only accepting state recognizing all ground instances of $R(x, g(a, y))$.

If $\tatoms$ is a function mapping linear atoms to a tree automata accepting the respective ground instances, then
the ground atoms generated by an IG $A/\{B_1,\ldots,$ $B_n\}$ are accepted by the tree automaton
$\tatoms(A) \cap \neg \tatoms(B_1) \cap \ldots \cap \neg \tatoms(B_n)$. Recall that tree automata
are closed under intersection ($\cap$) and complement ($\neg$), however the above
tree automaton may be exponentially larger in size compared to the size of $\tatoms(A)$ and the $\tatoms(B_i)$.

A \emph{linear} atom with membership constraint $A\colon x\in S$ is the pair of a linear atom $A$
and a constraint $x\in S$ where $x$ occurs in $A$ and $S$ is represented by a tree automaton.
It generates all ground instances $A\sigma$ where $x\sigma$ is accepted by the tree automaton
representing $S$. There is a function $\tatoms(A\colon x\in S)$ that computes in linear time
a tree automaton accepting exactly the generated ground instances of $A\colon x\in S$. Basically,
the state(s) representing the instances of $x$ in $A$ in $\tatoms(A)$ have to be replaced by the accepting states of
the tree automaton representing $S$.

Finally, tree automata can be
straight forwardly represented via MSLH clause sets. For example, the tree automaton representing
the ground instances of  $R(x, g(a, y))$ shown before, can be translated into the MSLH clause set
$\rightarrow Q_1(a)$, $\rightarrow Q_1(b)$, $Q_1(x), Q_1(y) \rightarrow Q_1(g(x,y))$,
$\rightarrow Q_2(a)$, $Q_2(x), Q_1(y) \rightarrow Q_3(g(x,y))$, and $Q_1(x), Q_3(y) \rightarrow Q_f(f_R(x,y))$.
This, together with Theorem~\ref{theorem:FiniteModelPropertyForMSLH}, eventually proves the following theorem.

\begin{theorem}
Linear disequality constraints (ADCs), linear implicit
generalizations (IGs) and linear atoms with membership constraints (AMCs) have the finite
model property.
\end{theorem}

This result can be easily generalized to any ``Boolean combination'' of
linear ADCs, IGs, and AMCs, because tree automata are closed under Boolean operations.
Our restriction on linearity does not imply that non-linear ADCs, IGs, and AMCs do \emph{not}
have the finite model property. This is an open problem. Non-linear MSLH clause sets do
not have the finite model property, because they are as expressive as full first-order logic.

%%%%%%%%%%%%%%%%%%%%%%%%%%%%%%%%%%%%%%%%%%%%%%
%%%%%%%%%%%%%%%%%%%%%%%%%%%%%%%%%%%%%%%%%%%%%%
%%%%%%%%%%%%%%%%%%%%%%%%%%%%%%%%%%%%%%%%%%%%%%
%%%%%%%%%%%%%%%%%%%%%%%%%%%%%%%%%%%%%%%%%%%%%%

\section{Model Finding by Approximation Refinement}\label{sec:mofind}

The approximation refinement approach~\cite{TeuckeW16}
cannot show satisfiability of the simple clause set with the two unit clauses\newline
\centerline{$R(x,x), \quad \neg R(y,g(y))$.}
The approximated clause set consisting of the three clauses\newline
\centerline{$T(f_R(x,y)), \quad \neg S(z) \lor \neg T(f_R(y,z)), \quad S(g(y))$}
immediately yields a refutation.
The problem is that $R(x,x) $ cannot be refined in such a way that all instances of the
conflict clause $\neg R(y,g(y))$ are excluded. 
The refinement loop instead ends up enumerating all $R(g^i(x),g^i(x))$ but
$R(g^{i+1}(y),g^{i+2}(y))$ will always remain as a conflict clause. 

The non-termination can be resolved, if the resolution inference in the
abstracted clause set can be blocked. Our suggestion in case of reflexive
relations is \emph{reflexive relation splitting}, i.e., we split a
reflexive relation into its reflexive part $R_\rfl$ and irreflexive part $R_\irr$.
For the example, this yields\newline
\centerline{$R_\rfl(x,x), \quad \neg R_\irr(y,g(y))$}
and after approximation\newline
\centerline{$T(f_{R_\rfl}(x,y)), \quad \neg S(z) \lor \neg T(f_{R_\irr}(y,z)), \quad S(g(y))$.}
Now the approximation is saturated. The operation preserves satisfiability
because the two $R$ literals could not be resolved anyway.

In general, for each predicate $R$ with a reflexivity axiom, 
all occurrences of atoms $R(s,t)$ are replaced with $R_\rfl(s,t)$ and/or $R_\irr(s,t)$.
If $s$ and $t$ are not unifiable,  $R(s,t)$ is replaced with $R_\irr(s,t)$.
If there is an mgu $\sigma$ of $s$ and $t$, $R(s,t)$ is replaced with both
$R_\rfl(s\sigma,t\sigma)$ and $R_\irr(s,t)$.
More precisely, any clause
$C \vee [\neg]R(s,t)$ is replaced by two clauses:
$C \vee [\neg]R_\irr(s,t)$ and $C\sigma \vee [\neg]R_\rfl(s\sigma,t\sigma)$.
The process is repeated until all atom occurrences with $R$ have been replaced.
In the final result, any clause that contains an atom of the form $R_\irr(s,s)$ can be deleted.

More formally, the following transition system 
replaces a reflexive $R$ by the two new predicates.

\smallskip\noindent
\shortrules{Irreflexive}{$N\disjun\{[\neg]R(s,t)\lor C\}$}{$N\disjun\{[\neg]R_\irr(s,t)\lor C\}$}{provided $s$ and $t$ are not unifiable}{RRS}{10.5}

\smallskip\noindent
\shortrules{Reflexive}{$N\disjun\{[\neg]R(s,t)\lor C\}$}{$N\disjun\{[\neg]R_\irr(s,t)\lor C, [\neg]R_\rfl(s\sigma,t\sigma)\lor C\sigma\}$}{provided $s$ and $t$ are unifiable by an mgu $\sigma$}{RRS}{10.5}

\smallskip\noindent
\shortrules{Delete}{$N\disjun\{[\neg]R_\irr(s,s)\lor C\}$}{$N$}{\vspace{-4ex}}{RRS}{10}

\begin{lemma}
$\Rightarrow_\text{RRS}$ is terminating and confluent.
\end{lemma}
\begin{proof}[Proof sketch]
Termination is easy to prove.
Each application of the rules {Irreflexive} or {Reflexive} reduces the multiset of the numbers of $R$-occurrences
in all clauses, and no new occurrences of $R$ are ever introduced when {Delete} is applied. 
Each application of the rule {Delete} reduces the number of occurrences of $R_\irr$.
Combining these two properties into a well-founded multi-set-based ordering completes the proof of termination.
For local confluence, the non-obvious case is a clause $[\neg] R(s,t)\lor [\neg] R(s',t') \lor C$
where $R(s,t)$ and $R(s',t')$ share variables, 
and without loss of generality, $s$ and $t$ are unifiable by the mgu $\sigma$.
Applying first the reflexive transformation to the first literal yields the two clauses
 $[\neg] R_\rfl(s\sigma,t\sigma)\lor [\neg] R(s'\sigma,t'\sigma) \lor C\sigma$ and
$[\neg] R_\irr(s,t)\lor [\neg] R(s',t') \lor C$. Now the interesting case is where $s'$, $t'$ are
unifiable but $s'\sigma$ and $t'\sigma$ are not. Then we get with the mgu $\tau$ of $s',t'$:
$[\neg] R_\rfl(s\sigma,t\sigma)\lor [\neg] R_\irr(s'\sigma,t'\sigma) \lor C\sigma$,
$[\neg] R_\irr(s,t)\lor [\neg] R_\irr(s',t') \lor C$, and 
$[\neg] R_\irr(s\tau,t\tau)\lor [\neg] R_\rfl(s'\tau,t'\tau) \lor C\tau$.
This is also exactly the result we get when starting with a translation of $[\neg] R(s',t')$:
if $s'\sigma,t'\sigma$ are not unifiable, then $s\tau,t\tau$ are not unifiable as well.
For otherwise, $s\tau\tau', t\tau\tau'$ for unifier $\tau'$ is an instance of $s\sigma,t\sigma$,
so $s'\sigma,t'\sigma$ must be unifiable as well, a contradiction to the above assumption.
All other cases are similar to this case.
By Newman's Lemma, termination plus local confluence implies confluence.
\end{proof}

Given any finite clause set $N$, we write $\rrs(N)$ to address the normal form of $N$ after exhaustively applying $\Rightarrow_\text{RRS}$.
Notice that any clause $D\in \rrs(N)$ is an
instance of a clause in $N$ with respect to the renaming of $R_\rfl$, $R_\irr$ with $R$.
Moreover, we use $\rrs(C)$ as shorthand for $\rrs(\{C\})$ for any clause $C$.

\begin{lemma}[Reflexive Relation Splitting] \label{lem:reflrelsplit}
Let $N$ be a finite clause set that does not contain the predicates $R_\rfl$ and $R_\irr$. 
$N$ is satisfiable if and only if $\rrs(N)$ is satisfiable.
\end{lemma}
\begin{proof}
Since $\Rightarrow_\text{RRS}$ is confluent, we may assume that $\rrs(N)$ has been 
derived from $N$ by $\Rightarrow_\text{RRS}$ using a strategy that applies the rule Delete 
with priority.
In other words, whenever there is some clause that may be deleted by Delete, we delete it at once.
For the rest of the proof, we incorporate this strategy into the rewrite system.
Hence, we pretend that $\Rightarrow_\text{RRS}$ comprises the following modified rules:

\smallskip\noindent
\shortrules{Irreflexive}{$N\disjun\{[\neg]R(s,t)\lor C\}$}{$N\disjun\{[\neg]R_\irr(s,t)\lor C\}$}{provided $s$ and $t$ are not unifiable}{RRS}{10.5}

\smallskip\noindent
\shortrules{Reflexive$_1$}{$N\disjun\{[\neg]R(s,t)\lor C\}$}{$N\disjun\{[\neg]R_\irr(s,t)\lor C, [\neg]R_\rfl(s\sigma,t\sigma)\lor C\sigma\}$}{provided $s$ and $t$ are distinct but unifiable by mgu $\sigma$ and $C\sigma$ does not contain any atom $R_\irr(s',s')$}{RRS}{10.5}

\smallskip\noindent
\shortrules{Reflexive$_2$}{$N\disjun\{[\neg]R(s,t)\lor C\}$}{$N\disjun\{[\neg]R_\irr(s,t)\lor C\}$}{provided $s$ and $t$ are distinct but unifiable by mgu $\sigma$ and $C\sigma$ contains some atom $R_\irr(s',s')$}{RRS}{12}

\smallskip\noindent
\shortrules{Reflexive$_3$}{$N\disjun\{[\neg]R(s,s)\lor C\}$}{$N\disjun\{[\neg]R_\rfl(s,s)\lor C\}$}{\vspace{-3ex}}{RRS}{12.5}

\smallskip\noindent
We now prove the following auxiliary result: \newline
\textit{Claim~I:}
Let $N_1, N_2$ be clause sets such that
$N_1 \Rightarrow_\text{RRS} N_2$.
Moreover, let $\I{}$ be any Herbrand interpretation such that \newline
(1) for every ground term $s$ we have $R_\rfl(s,s) \in \I{}$ if and only if $R(s,s) \in \I{}$, \newline
(2) for all ground terms $s, t$ and all $R_\rfl(s,t) \in \I{}$ we have $s=t$, \newline
(3) for all ground terms $s, t$ we have $R_\irr(s,t) \in \I{}$ if and only if $R(s,t) \in \I{}$.
Then, $\I{} \models N_1$ if and only if $\I{} \models N_2$. \newline
\textit{Proof of Claim~I:}
We distinguish three cases concerning the derivation step from $N_1$ to $N_2$, 
where the terms $s, t$ may be non-ground:
\begin{enumerate}
	\item[(a)] $N_1 = N' \uplus \{ C' \vee [\neg] R(s,t)\}$ and 
			$N_2 = N' \uplus \{ C' \vee [\neg] R_\irr(s,t)\}$,
	\item[(b)] $N_1 = N' \uplus \{ C' \vee [\neg] R(s,t) \}$ where $s, t$ are distinct but unifiable by the mgu~$\sigma$ and 
			$N_2 = N' \uplus \{ C' \vee [\neg] R_\irr(s,t), \; C'\sigma \vee [\neg] R_\rfl(s\sigma,t\sigma)\}$,
	\item[(c)] $N_1 = N' \uplus \{ C' \vee [\neg] R(s,s) \}$ and $N_2 = N' \uplus \{ C' \vee [\neg] R_\rfl(s,s) \}$.
\end{enumerate}

We first consider the ``only if'' direction.
Hence, we assume $\I{} \models N_1$ and, in particular, $\I{} \models C' \vee [\neg] R(s,t)$ in cases (a), (b) and (c) and $\I{} \models C' \vee [\neg] R(s,t)$ in case (d).

\smallskip
\noindent
Ad (a).
Suppose $\I{} \not\models N_2$.
Then, we must have $\I{} \not\models C'\rho \vee [\neg]R_\irr(s,t)\rho$ for some grounding substitution $\rho$.
This entails $\I{} \not\models C'\rho$.
The latter, because of $\I{} \models C'\rho \vee [\neg]R(s,t)\rho$, leads to $\I{} \models [\neg]R(s,t)\rho$.
Moreover, by definition of $\I{}$
, we have $R(s,t)\rho \in \I{}$ if and only if $R_\irr(s,t)\rho \in \I{}$.
But this contradicts $\I{} \not\models [\neg]R_\irr(s,t)\rho$.

\smallskip
\noindent
Ad (b).
Suppose $\I{} \not\models N_2$.
Then, we either have $\I{} \not\models C'\rho \vee [\neg]R_\irr(s,t)\rho$ for some grounding substitution $\rho$, or
$\I{} \not\models C'\sigma\rho \vee [\neg]R_\rfl(s\sigma,t\sigma)\rho$ for the mgu $\sigma$ of $s, t$ and for some grounding substitution $\rho$.
Consider the former case.
In analogy to (a), we derive $\I{} \not\models [\neg]R_\irr(s,t)\rho$ and $\I{} \models [\neg]R(s,t)\rho$.
By definition of $\I{}$, we have $R(s,t)\rho \in \I{}$ if and only if $R_\irr(s,t)\rho \in \I{}$ -- a contradiction.
Consider the latter case.
In analogy to (a), we derive $\I{} \not\models [\neg]R_\rfl(s\sigma,t\sigma)\rho$ and $\I{} \models [\neg]R(s\sigma,t\sigma)\rho$.
Since $\sigma$ unifies $s$ and $t$, we get $s\sigma\rho = t\sigma\rho$.
By definition of $\I{}$, this entails $R(s\sigma,t\sigma)\rho \in \I{}$ if and only if $R_\rfl(s\sigma,t\sigma)\rho \in \I{}$.
Again, this constitutes a contradiction.

\smallskip
\noindent
Ad (c).
Suppose $\I{} \not\models N_2$.
Then, we have $\I{} \not\models C'\rho \vee [\neg]R_\rfl(s,t)\rho$ for some grounding substitution $\rho$.
In analogy to (a), we derive $\I{} \not\models [\neg]R_\rfl(s,s)\rho$ and $\I{} \models [\neg]R(s,s)\rho$.
By definition of $\I{}$, this entails $R(s,s)\rho \in \I{}$ if and only if $R_\rfl(s,s)\rho \in \I{}$.
Again, this constitutes a contradiction.

%%%%%%%%%%%%%%%%%%%%
\smallskip
\noindent
Next, we consider the ``if'' direction, i.e.\ we assume $\I{}\models N_2$.

\smallskip
\noindent
Ad (a).
We have $\I{} \models C' \vee [\neg] R_\irr(s,t)$.
Suppose $\I{} \not\models N_1$.
Then, there must be some grounding substitution $\rho$ for which $\I{} \not\models C'\rho \vee [\neg] R(s,t)\rho$.
This entails $\I{} \not\models C'\rho$, which, because of $\I{} \models C'\rho \vee [\neg]R_\irr(s,t)\rho$, leads to $\I{} \models [\neg]R_\irr(s,t)\rho$.
By definition of $\I{}$, we then have $R_\irr(s,t)\rho \in \I{}$ if and only if $R(s,t)\rho \in \I{}$.
This yields a contradiction, as we have $\I{} \not\models [\neg]R(s,t)\rho$.

\smallskip
\noindent
Ad (b).
We have $\I{} \models C' \vee [\neg] R_\irr(s,t)$ and $\I{} \models C'\sigma \vee [\neg] R_\rfl(s\sigma,t\sigma)$.
Suppose $\I{} \not\models N_1$.
Again, we get $\I{} \not\models C'\rho \vee [\neg] R(s,t)\rho$ for some grounding $\rho$.
This entails $\I{} \not\models C'\rho$, which leads to $\I{} \models [\neg] R_\irr(s,t)\rho$.
By Assumption~(3), we get $R(s,t)\rho \in \I{}$ if and only if $R_\irr(s,t)\rho \in \I{}$.
But then, $\I{} \not\models [\neg] R(s,t)\rho$ contradicts $\I{} \models [\neg] R_\irr(s,t)\rho$.

\smallskip
\noindent
Ad (c).
We have $\I{} \models C' \vee [\neg] R_\rfl(s,s)$.
Suppose $\I{} \not\models N_1$.
Then, we get $\I{} \not\models [\neg] R(s,s)\rho$ for some grounding substitution $\rho$.
But this contradicts Assumption~(1), which entails that $R_\rfl (s\rho,s\rho) \in \I{}$ if and only if $R(s\rho,s\rho) \in \I{}$.

This finishes the proof of Claim~I.
\hfill$\Diamond$

Consider the derivation
$N \Rightarrow_\text{RRS} N_1 \Rightarrow_\text{RRS} \ldots \Rightarrow_\text{RRS} \rrs(N)$
based on the modified rewrite system with integrated, prioritized Delete.
Let $\I{}$ be a Herbrand model of $N$.
Since $N$ does not contain the predicates $R_\rfl$ and $R_\irr$, we can bring $\I{}$ into the shape that meets the conditions of Claim~I and still ensure that $\I{} \models N$.
It then follows that $\I{} \models N$, $\I{} \models N_1, \ldots, \I{} \models \rrs(N)$.
Symmetrically, let $\I{}$ be a Herbrand model of $\rrs(N)$.
Since $\rrs(N)$ does not contain the predicate $R$, we can reshape $\I{}$ so that Claim~I is applicable and $\I{}$ is still a model of $\rrs(N)$.
Then, we get $\I{} \models \rrs(N), \ldots, \I{}\models N$.
Consequently, $N$ is satisfiable if and only if $\rrs(N)$ is.
\end{proof}

Notice that the above lemma holds independently of the fact whether there is a reflexivity
clause in $N$ or not. Such a clause would, of course, also be  transformed by  $\Rightarrow_{\text{RRS}}$.

Let us take a look at an example that is a little bit more involved. 
Consider an equivalence relation $R$ with 
the respective axiom clauses.
\begin{equation*}
\begin{split}
	  &\imp R(x,x) \\
R(x,y) &\imp R(y,x) \\
R(x,y),R(y,z) &\imp R(x,z) \\
\end{split}
\end{equation*}
Applying $\Rightarrow_\text{RRS}$ exhaustively 
results in the clause set
\begin{equation*}
	\begin{split}
		&\imp R_\rfl(x,x) \\
		R_\irr(x,y) &\imp R_\irr(y,x) \\
		R_\rfl(x,x) &\imp R_\rfl(x,x) \\
		R_\irr(x,y),R_\irr(y,z) &\imp R_\irr(x,z) \\
		R_\irr(x,y),R_\irr(y,x) &\imp R_\rfl(x,x) \\
		R_\irr(x,y),R_\rfl(y,y) &\imp R_\irr(x,y) \\
		R_\rfl(x,x),R_\irr(x,z) &\imp R_\irr(x,z) \\
		R_\rfl(x,x),R_\rfl(x,x) &\imp R_\rfl(x,x). \\
	\end{split}
\end{equation*}
After removing redundant clauses, we are conveniently left with just
\begin{equation*}
\begin{split}
&\imp R_\rfl(x,x) \\
R_\irr(x,y) &\imp R_\irr(y,x) \\
R_\irr(x,y),R_\irr(y,z) &\imp R_\irr(x,z) \\
\end{split}
\end{equation*}
which are no longer trivialized by the linear approximation $T(f_R(x,y))$ of the reflexivity axiom.
See also the example in the introduction, Section~\ref{sec:intro}, for another application of reflexive relation splitting.

The rule Reflexive replaces a clause by two clauses and can, therefore, cause an 
exponential blow up in the number of generated clauses. However, this is only the case for a clause
with several occurrences $R(s_i,t_i)$ such that the respective term pairs are all simultaneously unifiable.
This situation can be detected and then reflexive relation splitting may not be efficiently applicable.
However, the above example on the equivalence relation $R$ shows that in the case of variable chains as they
occur in the transitivity axiom, all of the eventually generated clauses become redundant, except one:
$R_\irr(x,y),R_\irr(y,z) \imp R_\irr(x,z)$. We have integrated reflexive relation splitting into {\sc Spass}-AR~\cite{TeuckeW16}
and have run it on the overall TPTP~\cite{Sut09}. There is no example in TPTP~v.7.2.0 showing the exponential blow
up and the set of problems solved by {\sc Spass}-AR with reflexive relation splitting is strictly larger
than without.

Nevertheless, reflexive relation splitting is, of course, not sufficient to transform all problems
with inherently infinite models based on a (ir)reflexive relation into clause sets that can eventually 
be decided by AR
via MSLH clause sets. Consider a strict partial ordering without endpoints:
\begin{equation*}
\begin{split}
R(x,x) &\imp  \\
       & \imp R(x, g(x)) \\
R(x,y), R(y,z)  &\imp  R(x,z).\\ 
\end{split}
\end{equation*}
Reflexive relation splitting yields
\begin{equation*}
\begin{split}
R_\rfl(x,x) &\imp  \\
       & \imp R_\irr(x, g(x)) \\
R_\irr(x,y), R_\irr(y,z)  &\imp  R_\irr(x,z)\\ 
R_\irr(x,y), R_\irr(y,x)  &\imp  R_\rfl(x,x)\\ 
\end{split}
\end{equation*}
but after approximation, the abstraction refinement does not terminate on the example. The reason is
the approximation of the clause $R_\irr(x, g(x))$ into the two clauses $S(g(x))$ and $S(y)\imp T(f_{R_\irr}(x, y))$
where the property is lost that in any ground instance of  $R_\irr(x, g(x))$ the first argument
has one occurrence of $g$ less than the second. This was resolved in~\cite{Peltier09} by the 
use of tuple tree automata.

%%%%%%%%%%%%%%%%%%%%%%%%%%%%%%%%%%%%%%%%%%%%%%
%%%%%%%%%%%%%%%%%%%%%%%%%%%%%%%%%%%%%%%%%%%%%%
%%%%%%%%%%%%%%%%%%%%%%%%%%%%%%%%%%%%%%%%%%%%%%
%%%%%%%%%%%%%%%%%%%%%%%%%%%%%%%%%%%%%%%%%%%%%%

\section{Discussion} \label{sec:discussion}

We have shown that the MSLH clause fragment has the finite model property
and can therefore not represent models of clause sets with inherently infinite models.
This applies to the model representation building blocks atoms with disequality constraints,
implicit generalizations, and atoms with membership constraints as well, if atoms
and terms are linear. For non-linear terms, our finite model property proof breaks,
and, in fact, the example from the introduction shows already that non-linear atoms
can represent models for clause sets with inherently infinite models.

Unsatisfiability of monadic shallow Horn clause sets is undecidable.
One occurrence of a clause $\Gamma \imp S(f(x,x))$ suffices to this end.
This can be seen by a respective monadic reformulation of the PCP encoding from the introduction.
On the other hand, models represented by ground instances of finite sets of (linear or non-linear)
atoms are also restricted in expressivity, because they cannot express any recursive structure.
For example, MSLH clause sets and extensions thereof have been successfully used
for the analysis of security protocols~\cite{Weidenbach99cade,Blanchet04} where (counter-) models cannot 
be expressed by ground instances of finite sets of atoms. In summary, and not surprisingly,
there is currently no unique superior model representation formalism.

If models are eventually constructed through the reversal of an approximation, the used representation
may have the finite model property and can still show satisfiability of clause sets
with inherently infinite models. We obtained this result via reflexive relation splitting.
This insight is already a consequence of \cite{Peltier09}. There, an approximation into
a theory of tuple tree automata is described and it is even complete with respect to models
generated out of these automata. We can currently not provide such a completeness
result although this would be highly desirable. On the other hand, our techniques are embedded
into a refutationally complete procedure, whereas the approach in \cite{Peltier09}
can only show satisfiability.

%%%%%%%%%%%%%%%%%%%%%%%%%%%%%%%%%%%%%%%%%%%%%%
%%%%%%%%%%%%%%%%%%%%%%%%%%%%%%%%%%%%%%%%%%%%%%
%%%%%%%%%%%%%%%%%%%%%%%%%%%%%%%%%%%%%%%%%%%%%%
%%%%%%%%%%%%%%%%%%%%%%%%%%%%%%%%%%%%%%%%%%%%%%

\end{document}